\documentclass{article}

\PassOptionsToPackage{table}{xcolor}
\usepackage{iclr2027_conference,times}
\usepackage{amsmath,amssymb,amsthm}
\usepackage{booktabs}
\usepackage{array}
\usepackage{graphicx}
\usepackage{float}
\usepackage{microtype}
\usepackage{xcolor}
\usepackage{hyperref}
\usepackage{url}
\hypersetup{colorlinks=true,citecolor=blue!55!black,linkcolor=blue!55!black,urlcolor=blue!55!black}
\hypersetup{pdfauthor={Maojun Sun, Yancheng Yuan, Jian Huang, Ruijian Han},
  pdftitle={When Does Dense Retrieval Need Asymmetric Geometry? A Bias-Variance Theory of Shared and Dual Projections}}
\usepackage{multirow}

\graphicspath{{./figures/}}
\newtheorem{theorem}{Theorem}
\newtheorem{proposition}{Proposition}
\newtheorem{corollary}{Corollary}

\theoremstyle{remark}

\newcommand{\Ms}{\mathcal M_s}
\newcommand{\Md}{\mathcal M_d}
\newcommand{\Ss}{\mathcal S_r^+}
\newcommand{\Dd}{\mathcal D_r}
\newcommand{\E}{\mathbb E}
\newcommand{\norm}[1]{\lVert #1\rVert}
\newcommand{\ip}[2]{\langle #1,#2\rangle}
\newcommand{\sym}{\operatorname{sym}}
\newcommand{\rank}{\operatorname{rank}}
\newcommand{\tr}{\operatorname{tr}}

\title{When Does Dense Retrieval Need Asymmetric Geometry?\\
A Bias--Variance Theory of Shared and Dual Projections}

\author{Maojun Sun$^{1}$, Yancheng Yuan$^{1,\dagger}$, Jian Huang$^{1,\dagger}$, Ruijian Han$^{1,\dagger}$\\[3pt]
$^{1}$ The Hong Kong Polytechnic University\\
$^{\dagger}$ Joint corresponding author}

\iclrfinalcopy
\begin{document}
\maketitle
\lhead{Preprint}

\begin{abstract}
Dense retrieval powers retrieval-augmented generation, semantic search, and question answering, yet the theoretical basis for choosing between shared and dual query–document projections remains unclear. We introduce a bias–variance theory for low-rank bilinear scoring. Shared projections induce positive-semidefinite operators, whereas dual projections realize arbitrary low-rank operators. We derive their exact approximation gap and prove a local Gaussian boundary: dual has lower risk exactly when squared directional signal exceeds the estimation cost of its additional degrees of freedom. This boundary motivates the Cross-fitted Asymmetry Risk Selector (CARS), which estimates reproducible directional signal from training pairs; its Gaussian counterpart admits exact selection-power and regret formulas. Guided by the theory, we run retrieval experiments across multiple datasets and embedding models. The mean Dual-minus-Shared NDCG@10 advantage more than doubles as query rotation increases from \(0^\circ\) to \(90^\circ\). In the rank–sample-size grids, Shared wins 13 of 16 cells at \(n=32\), whereas Dual wins all 32 cells at \(n=1024\) and \(2048\). Consistent with this shift, all 168 comparable operator-risk curves move toward Dual as training data grow. Compared to the two fixed-geometry baselines, CARS reduces held-out regret by 49–96\% and achieves 90.1\% mean geometry-selection accuracy.
\end{abstract}

\section{Introduction}

Dense retrieval supplies evidence to question-answering and retrieval-augmented
generation systems and powers semantic search
\citep{karpukhin2020dpr,lewis2020rag,thakur2021beir}.  When adapting frozen
query and document embeddings in $\mathbb R^p$, one can apply the \emph{same}
rank-$r$ projection to both (\emph{shared}) or fit two projections
(\emph{dual}), where $1\le r\le p$. Throughout, “shared” applies the same projection to queries and documents, while “dual” uses separately parameterized projections to model directional differences between them. Their finite-sample tradeoff depends on signal strength and estimation noise. Figure~\ref{fig:share-vs-dual} illustrates the two paradigms.

\begin{figure}[htbp]
    \centering
    \includegraphics[width=\textwidth]{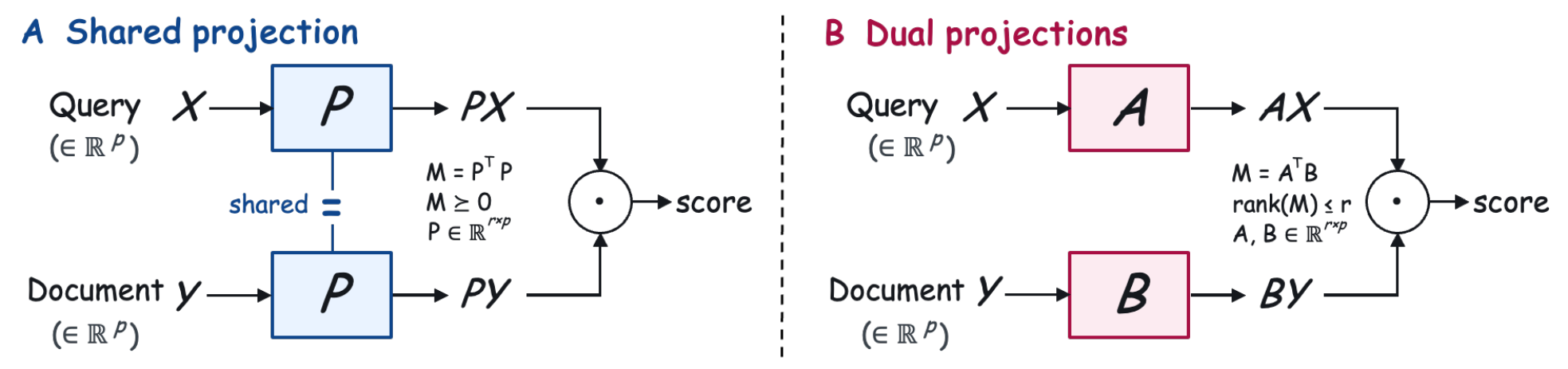}
    \caption{Shared and dual projections for dense retrieval. A shared map \(P\) induces a PSD operator \(M=P^\top P\), while dual maps \(A\) and \(B\) induce a general operator \(M=A^\top B\) of rank at most \(r\).}
    \label{fig:share-vs-dual}
\end{figure}

The key distinction is geometric. With inner-product scoring, a shared
projection $P\in\mathbb R^{r\times p}$ gives the similarity matrix
$M=P^\top P$, which is positive semidefinite (PSD) because
$v^\top Mv=\|Pv\|_2^2\ge0$ for every $v$. Separate query and document
projections $A,B\in\mathbb R^{r\times p}$ instead give $M=A^\top B$,
which can represent any rank-at-most-$r$ matrix and need not be symmetric.
This flexibility reduces approximation bias but opens more directions to
training noise. We ask \emph{when directional signal pays for those extra
directions}, and whether the answer can be estimated from training data.

Our contributions are summarized as follows. (i) We characterize the operator classes induced by shared and dual rank-\(r\) projections and derive the exact approximation loss imposed by shared projections. (ii) We establish a local Gaussian bias–variance boundary that quantifies when directional signal outweighs the additional estimation cost of dual projections. (iii) We derive a Stein-unbiased rule with exact selection power and regret in the local Gaussian model, then introduce CARS as a cross-fitted selector for real embeddings. (iv) We examine the predicted signal and sample-size effects through controlled simulations, full-corpus retrieval, and held-out operator-risk experiments on real embeddings and datasets.

\paragraph{Related Work.}
Shared PSD metric learning and unconstrained bilinear similarity learning
study the two scoring families
\citep{weinberger2009lmnn,davis2007itml,chechik2010oasis}.
Dense-retrieval work compares shared and asymmetric dual encoders and
develops methods for adapting pretrained embeddings
\citep{dong2022dual,yoon2024search,maekawa2026era}.
CCA, low-rank geometry, and Stein risk estimation provide related
mathematical tools
\citep{hotelling1936cca,bach2005pcca,edelman1998geometry,candes2013sure,stein1981sure}.
What remains unresolved is a theoretical criterion for when the flexibility
of dual projections justifies their additional estimation error.
We establish an exact boundary in the local Gaussian operator-risk model:
Dual has lower risk precisely when its squared directional signal exceeds
the variance cost of its additional degrees of freedom.
We then develop risk-based selection rules and examine this tradeoff in
retrieval and held-out operator-risk experiments.

\section{Operator geometry and approximation}

Let $X,Y\in\mathbb R^p$ denote frozen query and document embeddings, and
let $Y^+,Y^-$ be a positive and a negative document for query $X$.
With $D=Y^+-Y^-$, the bilinear margin associated with
$M\in\mathbb R^{p\times p}$ is
\[
m_M(X,Y^+,Y^-)=X^\top MD=\ip{M}{XD^\top}_F.
\]
Here $\langle\cdot,\cdot\rangle_F$ and $\|\cdot\|_F$ denote the
Frobenius inner product and norm.
The two rank-constrained families are
\begin{align*}
\Ms(r)&=\{M=M^\top,\ M\succeq0,\ \rank(M)\le r\},\\
\Md(r)&=\{M:\rank(M)\le r\}.
\end{align*}
A shared projection induces $M=P^\top P\in\Ms(r)$; dual projections induce
$M=A^\top B\in\Md(r)$.  Conversely, if
$M=U\Sigma V^\top$ has rank at most $r$, padded factors
$A=\Sigma^{1/2}U^\top$ and $B=\Sigma^{1/2}V^\top$ satisfy $A^\top B=M$.
Thus $\Ms(r)\subsetneq\Md(r)$.

Let a population target be $M_\star=H+K$, where
\[
H=\frac{M_\star+M_\star^\top}{2},\qquad
K=\frac{M_\star-M_\star^\top}{2}.
\]
Write $H=U\operatorname{diag}(\eta_1,\ldots,\eta_p)U^\top$ with
$\eta_1\ge\cdots\ge\eta_p$, and let $J_r$ index the largest at most $r$ positive eigenvalues.
We first quantify the approximation error of requiring the scoring operator to be shared.

\begin{theorem}[Exact shared approximation]\label{thm:projection}
A Frobenius-nearest member of $\Ms(r)$ is
$
S_r^+=U\operatorname{diag}(\eta_j\mathbf 1\{j\in J_r\})U^\top,
$
and the minimum squared error is
\begin{equation}
\norm K_F^2+
\sum_{\eta_j<0}\eta_j^2+
\sum_{\substack{\eta_j>0, j\notin J_r}}\eta_j^2.
\label{eq:projection}
\end{equation}
The minimum over $\Md(r)$ is $\sum_{j>r}\sigma_j(M_\star)^2$;
the exact Shared-minus-Dual approximation gap is
Equation~\eqref{eq:projection} minus this singular-value tail.
\end{theorem}

\begin{figure}[htbp]
\centering
\includegraphics[width=\linewidth]{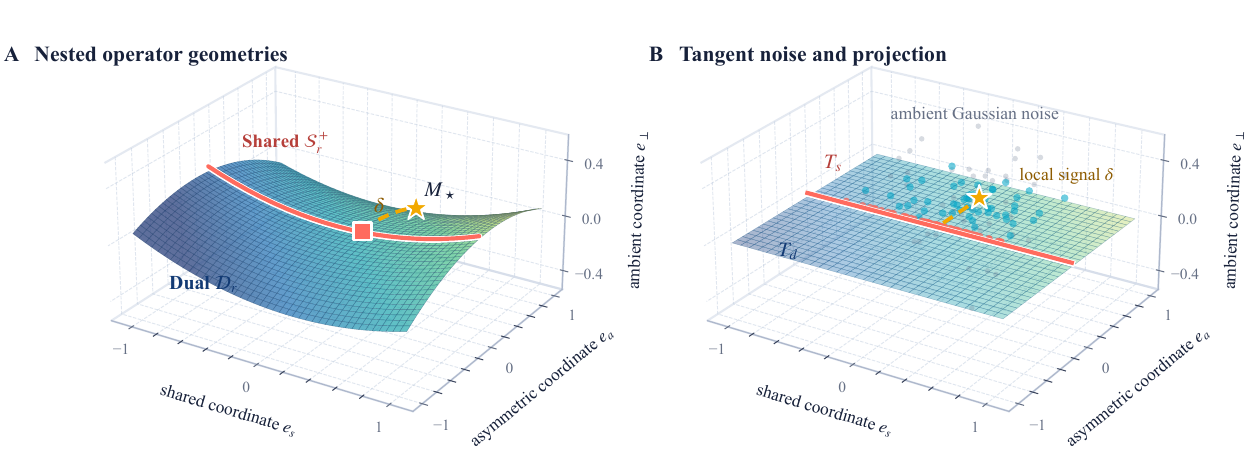}
\caption{Geometry behind the boundary. {\textbf{A:} the shared PSD
family lies inside the dual family; $\delta$ schematically marks a target's
departure from sharing.} \textbf{B:} locally, sharing
removes $k$ dual-only noise coordinates but also discards signal along them;
$k$ is the extra dimension derived in Section~\ref{sec:boundary}.}
\label{fig:manifold-geometry}
\end{figure}

The three terms in \eqref{eq:projection} are skew energy, negative eigenvalues, and discarded positive
eigenvalues. Appendix~\ref{app:projection} gives the proof.
Theorem~\ref{thm:projection} concerns approximation with a known target. Figure~\ref{fig:manifold-geometry}A illustrates this nesting: a target outside the shared family can be represented more accurately by dual projections.

\paragraph{Retrieval interpretation.}
For centered query--positive-document pairs $(X,Y^+)$, let
$\Sigma_x,\Sigma_y$ be their covariances, and let $Y^-$ be an independent
centered document with covariance $\Sigma_y$. The relevance moment is
$C=\E[X(Y^+-Y^-)^\top]=\E[XY^{+\top}]$. After whitening the two views,
take $M_\star=T=\Sigma_x^{-1/2}C\Sigma_y^{-1/2}$ in
Theorem~\ref{thm:projection}. Its Shared-minus-Dual approximation gap
then equals the gap in squared optimal rank-$r$ positive--negative
separation under a unit negative-score second-moment constraint.
Appendix~\ref{app:deflection} proves this equivalence and gives a
rotation example in which the population advantage of Dual grows with
query--document mismatch. With finitely many training pairs, however,
Dual can also fit noise in its extra directions. We next ask whether a
distribution-free bound quantifies that estimation cost.

\begin{proposition}[Global complexity bound]
\label{prop:global}
For $n$ training triples, let $D_i=Y_i^+-Y_i^-$ and
$W_i=X_iD_i^\top$. Suppose $\norm{W_i}_F\le R$, and restrict
$\norm M_F\le B$. For $j\in\{s,d\}$, define the empirical Rademacher
complexity of the corresponding linear-score class as
\[
\mathfrak R_{n,j}=\E_\epsilon\sup_{\substack{M\in\mathcal M_j(r)\\
\norm M_F\le B}}
\ip{M}{\frac1n\sum_{i=1}^n\epsilon_iW_i}_F.
\]
The $\epsilon_i$ are independent uniform signs in $\{-1,+1\}$.
If $Q=n^{-1}\sum_i\epsilon_iW_i$ and $Q_r$ is its rank-$r$ SVD truncation,
then $\mathfrak R_{n,s}\le\mathfrak R_{n,d}=B\E_\epsilon\norm{Q_r}_F\le{BR}/{\sqrt n}.$
\end{proposition}

Proposition~\ref{prop:global} preserves the ordering $\mathfrak R_{n,s}\le\mathfrak R_{n,d}$, but bounds both by the same coarse ceiling $BR/\sqrt n$. This ceiling does not quantify the extra estimation cost of dual's additional directions. The local noise model in Section~\ref{sec:boundary} resolves that difference; Appendix~\ref{app:global} gives the proof and a Lipschitz-loss extension.

\section{The bias--variance boundary}\label{sec:boundary}

To measure that cost, we consider a local Gaussian experiment
around a rank-$r$ shared operator: an observed population operator is
perturbed by isotropic noise of scale $\sigma/\sqrt n$, while directional
departure from the shared family is also of order $n^{-1/2}$. This common
scale makes approximation gain and estimation cost directly comparable.

Define the regular rank-$r$ manifolds
\[
\Dd=\{M:\rank(M)=r\},\qquad
\Ss=\{S=S^\top\succeq0:\rank(S)=r\}.
\]
Let $S=U\Lambda U^\top\in\Ss$, where the columns of $U\in\mathbb R^{p\times r}$
are orthonormal and $\Lambda\in\mathbb R^{r\times r}$ is positive definite.
Let $T_d$ and $T_s$ denote the tangent spaces to $\Dd$ and $\Ss$ at $S$,
respectively, and let $\Pi_T$ denote Frobenius-orthogonal projection onto $T$.
Assume that $M_n\in\Dd$ is a sequence of local alternatives with
$M_n=S+n^{-1/2}H+O(n^{-1})$ for fixed $H\in T_d$. We observe
\[
Y_n=M_n+\frac{\sigma}{\sqrt n}G_n,
\]
where $\sigma>0$ and $G_n\in\mathbb R^{p\times p}$ has independent
standard Gaussian entries. In retrieval, \(Y_n\) is a local Gaussian model for the empirical relevance
moment \(\widehat C_n=n^{-1}\sum_{i=1}^n X_iD_i^\top\), while \(M_n\)
represents its rank-\(r\) population target.
Let $\widehat M_d$ and $\widehat M_s$ be
Frobenius-nearest points to $Y_n$ in the closed families $\Md(r)$ and
$\Ms(r)$, respectively.

\begin{theorem}[Local estimation cost]\label{thm:localrisk}
In the local model above, $T_s\subset T_d$, $\dim T_d=2pr-r^2$, and
$\dim T_s=pr-r(r-1)/2$. The number of dual-only tangent directions is
\begin{equation}
k=\dim T_d-\dim T_s=\frac{r(2p-r-1)}{2}.
\label{eq:k}
\end{equation}
Moreover, as $n\to\infty$,
\begin{align*}
n\E\norm{\widehat M_d-M_n}_F^2&\to\sigma^2\dim T_d,\\
n\E\norm{\widehat M_s-M_n}_F^2&\to\norm{\Pi_{T_s^\perp}H}_F^2+\sigma^2\dim T_s.
\end{align*}
\end{theorem}

Let $\mathcal A=T_d\cap T_s^\perp$ denote the dual-only directions, so
$\dim\mathcal A=k$. The extra dimension decomposes as
$k=r(r-1)/2+r(p-r)$: the first term corresponds to skew perturbations within
the active $r$-dimensional subspace, while the second allows the left and
right cross-subspace perturbations to differ.

\begin{corollary}[Local phase boundary]\label{cor:phase}
Let $\delta^2=\norm{\Pi_{\mathcal A}H}_F^2$. The dual estimator has lower local
asymptotic risk than the shared estimator if and only if
\begin{equation}
\delta^2>\sigma^2k.\label{eq:oracleboundary}
\end{equation}
For
$\Delta_n=M_n-S=n^{-1/2}H+O(n^{-1})$, $n\norm{\Pi_{\mathcal A}\Delta_n}_F^2
\rightarrow\delta^2$. Away from equality, the same first-order
decision compares $n\norm{\Pi_{\mathcal A}\Delta_n}_F^2$ with $\sigma^2k$.
\end{corollary}

The boundary prices the $k$ additional noisy coordinates against
the signal they can capture. Figure~\ref{fig:manifold-geometry}B depicts
this local tradeoff: sharing removes variance in $\mathcal A$ but also
discards the signal there. The synthetic crossing in
Figure~\ref{fig:phase} tests the corresponding signal--sample-size
prediction.

\subsection{A selection rule from noisy data}

Let \(Z_n=\sqrt n\,\Pi_{T_d}(Y_n-S)\). Since \(H\in T_d\), \(Z_n=H+\sigma\,\Pi_{T_d}G_n+O(n^{-1/2})\), so \(Z_n\) converges in distribution to \(Z=H+\sigma G_d\), where \(G_d\) is standard Gaussian noise in \(T_d\). We state the selection rule for this limit model, in which the risks of the two projections match Theorem~\ref{thm:localrisk}.

The boundary in Corollary~\ref{cor:phase} depends on the unknown
\(\delta^2\). Its plug-in estimate \(\norm{\Pi_{\mathcal A}Z}_F^2\) is
biased upward, since its expectation is \(\delta^2+\sigma^2k\). We
correct this bias with Stein's unbiased risk estimate (SURE)
\citep{stein1981sure}, treating \(\sigma\) as known. For
\(j\in\{s,d\}\), the estimator \(\widehat H_j=\Pi_{T_j}Z\) has risk
\(R_j(H)=\E\norm{\widehat H_j-H}_F^2\), and
\[
\widehat R_j=\norm{Z-\Pi_{T_j}Z}_F^2+\sigma^2\bigl(2\dim T_j-\dim T_d\bigr)
\]
is an unbiased estimate of \(R_j(H)\). Let \(\widehat\pi\in\{s,d\}\)
select the family with the smaller \(\widehat R_j\), with ties resolved
in favor of \(s\).

\begin{theorem}[SURE selection and model-choice regret]
\label{thm:sure}
In the model above, \(R_s(H)-R_d(H)=\delta^2-\sigma^2k\), and
\(\widehat\pi=d\) if and only if
\begin{equation}
\norm{\Pi_{\mathcal A}Z}_F^2>2\sigma^2k.
\label{eq:surethreshold}
\end{equation}
Moreover,
\begin{equation}
\Pr(\widehat\pi=d)=\Pr\{\chi_k^2(\lambda)>2k\},
\qquad \lambda=\delta^2/\sigma^2,
\label{eq:power}
\end{equation}
where \(\chi_k^2(\lambda)\) is a noncentral chi-squared variable with
\(k\) degrees of freedom and noncentrality \(\lambda\). Define the
model-choice regret as \(R_{\widehat\pi}(H)-\min\{R_s(H),R_d(H)\}\).
Its expectation is
\begin{equation}
|\delta^2-\sigma^2k|\times
\begin{cases}
\Pr\{\chi_k^2(\lambda)>2k\}, & \delta^2<\sigma^2k,\\
\Pr\{\chi_k^2(\lambda)\le 2k\}, & \delta^2>\sigma^2k,
\end{cases}
\label{eq:regret}
\end{equation}
and it is zero when \(\delta^2=\sigma^2k\).
\end{theorem}

The threshold \(2\sigma^2k\) combines two terms: one \(\sigma^2k\)
removes the expected noise energy in \(\mathcal A\), and the other is
the boundary itself. Thus, unlike the plug-in rule
\(\norm{\Pi_{\mathcal A}Z}_F^2>\sigma^2k\), SURE rarely selects dual
without dual-only signal: when \(\delta=0\), the probability is at most
\((2/e)^{k/2}\). Model-choice regret weights the risk gap by the
probability of selecting the worse procedure; it is not the
post-selection risk of \(\widehat H_{\widehat\pi}\).

\subsection{CARS: Cross-fitted Asymmetry Risk Selector}\label{sec:cars}

In the local Gaussian experiment, SURE gives an unbiased estimate of
the Shared--Dual risk difference when $S$, $T_s$, $T_d$, and $\sigma$
are specified. For real embeddings, the population geometry and noise
covariance are unknown. We therefore use labeled training triples to
measure how consistently the dual--shared fit difference appears across
disjoint training halves, while penalizing its sampling variation.
For a training subset $I$, define the relevance moment and the
dual--shared fit difference
\[
\widehat C_I=\frac{1}{|I|}\sum_{i\in I}X_i(Y_i^+-Y_i^-)^\top,
\qquad
V_I=\mathcal P_{d,r}(\widehat C_I)-\mathcal P_{s,r}(\widehat C_I),
\]
where $\mathcal P_{d,r}$ is rank-$r$ SVD truncation and
$\mathcal P_{s,r}$ retains the largest $r$ positive eigenvalues of the
symmetric part, as in Theorem~\ref{thm:projection}. For each of $L$
random disjoint half-splits $(I_{b,1},I_{b,2})$ of the training set, the
\emph{Cross-fitted Asymmetry Risk Selector} (CARS) computes
\begin{equation}
\widehat\Gamma_{\rm CARS}=
\frac1L\sum_{b=1}^L\left[
\ip{V_{I_{b,1}}}{V_{I_{b,2}}}_F
-\frac14\norm{V_{I_{b,1}}-V_{I_{b,2}}}_F^2\right],
\label{eq:cars-score}
\end{equation}
and selects Dual if $\widehat\Gamma_{\rm CARS}>0$, Shared otherwise.
The cross-half inner product retains structure that replicates across
splits; the disagreement term prices its sampling variation. Under the
independent-half model below, the split difference has expected squared norm
$4\tr(\Sigma)/n$, so the factor $1/4$ charges the full-sample variance
cost $\tr(\Sigma)/n$.

This score has a direct connection to the local boundary.
Suppose a split residual satisfies $V_{I_{b,j}}=\Delta+\varepsilon_{b,j}$,
with independent centered half-sample errors of covariance $2\Sigma/n$.
Then
\begin{equation*}
\E\widehat\Gamma_{\rm CARS}
=\norm{\Delta}_F^2-\frac{\tr(\Sigma)}{n}.
\end{equation*}
If the residual lies in $\mathcal A$, $\Delta=n^{-1/2}\Pi_{\mathcal A}H$
to first order, and $\Sigma=\sigma^2I_k$, the right-hand side is
$(\delta^2-\sigma^2k)/n$. Thus split agreement and disagreement estimate
the same signal--variance comparison as Corollary~\ref{cor:phase}, using
training data rather than a specified noise variance. The expectation
identity is proved in Appendix~\ref{app:cars-score}; the held-out
operator-loss study in RQ4 evaluates its geometry choices on real
embeddings.

\section{Experimental Setting}\label{sec:experimental-setting}

The theory yields a sequence of empirical questions: whether its
signal--uncertainty crossing is visible, whether the signal can be estimated,
whether directional mismatch changes retrieval rankings, and whether these
estimates improve geometry choice on real embeddings. In real-data Shared--Dual comparisons, both
families use paired training subsets and the same frozen embeddings,
relevance labels, and test queries. The four base encoders are E5-base-v2,
BGE-base, GTE-base, and
Contriever-MSMARCO
\citep{wang2022e5,xiao2023cpack,li2023gte,izacard2022contriever};
the real-embedding studies use the subsets specified below. Retrieval
quality is measured by mean test-query NDCG@10, defined in
Appendix~\ref{app:experiments}.

\paragraph{RQ1: Does the predicted boundary appear?}
We test whether the predicted signal--sample-size crossover appears in
controlled two-view retrieval. A Gaussian retrieval experiment rotates the
document loading in dimension $p=32$ and
fits ranks $4/8/16$ using $32$--$2048$ training pairs. Each positive is
ranked against 299 independent negatives, with five paired seeds at every
mismatch--sample-size setting.

\paragraph{RQ2: Can directional signal be estimated?}
We test whether correcting estimation noise makes observed asymmetry more
informative about the held-out benefit of Dual. In the rank-8 simulation,
known population PSD distance provides a target for comparing raw plug-in
asymmetry with its uncertainty-corrected estimate. On real embeddings, we
compare raw and corrected training scores across 840 operator fits on ten
tasks. Disjoint held-out relevance moments define the Shared--Dual
operator-loss advantage used to assess those scores.

\paragraph{RQ3: How do mismatch, rank, and training size affect retrieval?}
We test whether directional mismatch changes the relative retrieval quality of
Shared and Dual projections. Starting from frozen real embeddings, we rotate
only query vectors in eight relevance-informed planes through seven angles
from $0^\circ$ to $90^\circ$, leaving document embeddings, evaluation corpora,
and relevance labels unchanged. We train rank-16 adapters with 1,024 queries
on FEVER, HotpotQA, NQ, MS MARCO, and CQA-TeX using BGE-base and GTE-base.
Each angle uses three paired seeds; learning rates and checkpoints are
selected on validation queries. To examine the interaction with model
capacity and data availability, we also vary rank over $\{4,8,16,32\}$ and
training size from 32 to 2,048 queries. Figure~\ref{fig:rank-sweep} shows
four complete corpus--encoder grids, with three paired seeds per cell. We rank each complete document
collection by exact inner-product search.
Figures~\ref{fig:controlled-real} and~\ref{fig:rank-sweep} report these
retrieval comparisons.

\paragraph{RQ4: Does operator risk shift with more data, and can CARS select the better geometry?}
We test whether more training queries shift held-out operator fit toward
Dual and whether CARS can choose the lower-risk geometry using training data
alone. The displayed operator-fit curves use NQ and SciFact with all four
encoders and three seeds across nested training sizes and ranks. A disjoint
query set defines the relevance-moment target
$M_{\rm test}$; for geometry $g$, normalized loss is
$L_g=\|\widehat M_g-M_{\rm test}\|_F^2/\|M_{\rm test}\|_F^2$, so positive
$L_s-L_d$ favors Dual. Figure~\ref{fig:gpu-operator} displays the seven
measured training sizes for NQ (ranks $4/8/16/32$, $n=256$--$2071$) and
SciFact (ranks $64/128/256$, $n=256$--$647$).

For geometry selection, we evaluate CARS (Section~\ref{sec:cars}) on five disjoint held-out query
folds of ArguAna, MS MARCO, FEVER, NQ, and SciFact, using all four encoders,
nested training sizes, ranks $4/8/16/32$, and 20 repeated internal
half-splits. For each sample-size--rank cell, selection regret is
$L_{\widehat g}-\min(L_s,L_d)$; accuracy is the fraction of cells choosing
the lower-loss geometry, assigning ties to Shared. Regret is averaged over
cells within each held-out fold and then over the five folds. A strict fold
win requires lower fold-mean regret than both fixed rules; ties win nothing.

\section{Results}\label{sec:results}

\subsection{RQ1: Does the predicted boundary appear?}

\begin{figure}[htbp]
\centering
\includegraphics[width=\linewidth]{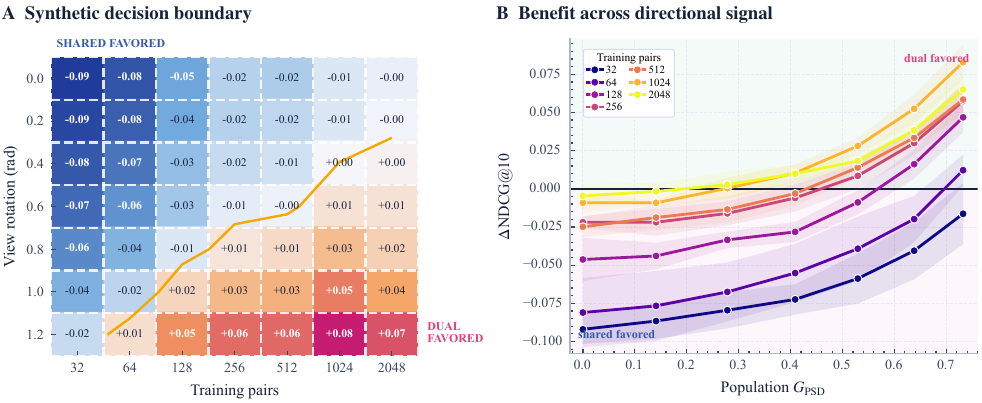}
\caption{A Shared-to-Dual crossover in controlled two-view retrieval.
\textbf{A:} Mean Dual-minus-Shared NDCG@10 across mismatch and training size; the gold contour marks zero.
\textbf{B:} The same gap versus population distance $G_{\rm PSD}$ to the shared PSD family; bands show 95\% normal confidence intervals.}
\label{fig:phase}
\end{figure}

In Figure~\ref{fig:phase}A, the first positive rank-8 grid mean occurs at
$n=256$ for $0.8$-radian rotation but at $n=64$ for $1.2$-radian rotation.
Figure~\ref{fig:phase}B likewise shows a growing Dual advantage as population
distance from the PSD family increases. Together, the panels show the
qualitative crossing predicted by Corollary~\ref{cor:phase}:
stronger directional mismatch needs fewer training pairs to overcome
Dual's estimation cost. The zero contour shows that neither geometry
dominates throughout the grid. A separate local-Gaussian calibration
matches the risk limits in Theorem~\ref{thm:localrisk} within $0.25\%$
and the SURE selection probabilities in Theorem~\ref{thm:sure} within
$0.001$ (Appendix Table~\ref{tab:local-calibration}).

\subsection{RQ2: Can directional signal be estimated?}

The plug-in bias preceding Theorem~\ref{thm:sure} predicts that a
large fitted asymmetry can reflect sampling noise. Figure~\ref{fig:noise-correction}A
shows this inflation directly. On real embeddings, the raw plug-in trend
is nearly flat across score deciles, whereas the corrected-score trend in
Figure~\ref{fig:noise-correction}B rises from Shared-favored to
Dual-favored held-out fit. Accounting for variation between training splits, as in Equation~\eqref{eq:cars-score}, makes the corrected score more informative about which family has lower operator loss on held-out
queries.

\begin{figure}[htbp]
\centering
\includegraphics[width=\linewidth]{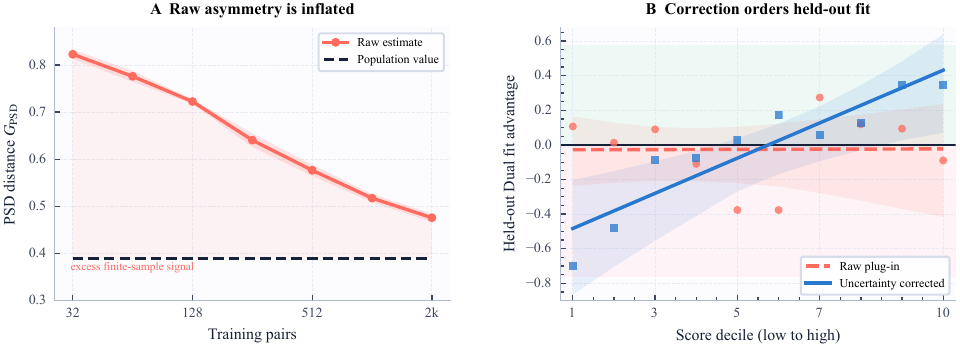}
\caption{Raw asymmetry is inflated, while correction better predicts held-out fit.
\textbf{A:} Training PSD distance exceeds its population value; shading shows a 95\% $t$ interval.
\textbf{B:} Held-out Shared-minus-Dual operator-loss advantage across raw and corrected score deciles. Positive values favor Dual.}
\label{fig:noise-correction}
\end{figure}

\subsection{RQ3: How do mismatch, rank, and training size affect retrieval?}\label{sec:rq3-results}
The rank-one rotation example in Appendix~\ref{app:deflection} predicts a
widening optimal-separation gap as cross-view mismatch increases. The
following experiment asks whether an analogous trend appears in trained,
full-corpus retrieval.
Figure~\ref{fig:controlled-real} reports the test NDCG@10 gap at each query
rotation. Averaged across the five datasets, the Dual--Shared gap rises
from about $.0032$ at $0^\circ$ to $.0146$ at $90^\circ$, a $4.6$-fold increase.
All five dataset means increase between these endpoints; on MS MARCO, the
gap changes sign from about $-.0008$ to $+.0219$.
The widening gap is consistent with the predicted benefit of Dual
projections under stronger cross-view mismatch.

\begin{figure}[htbp]
\centering
\includegraphics[width=\linewidth]{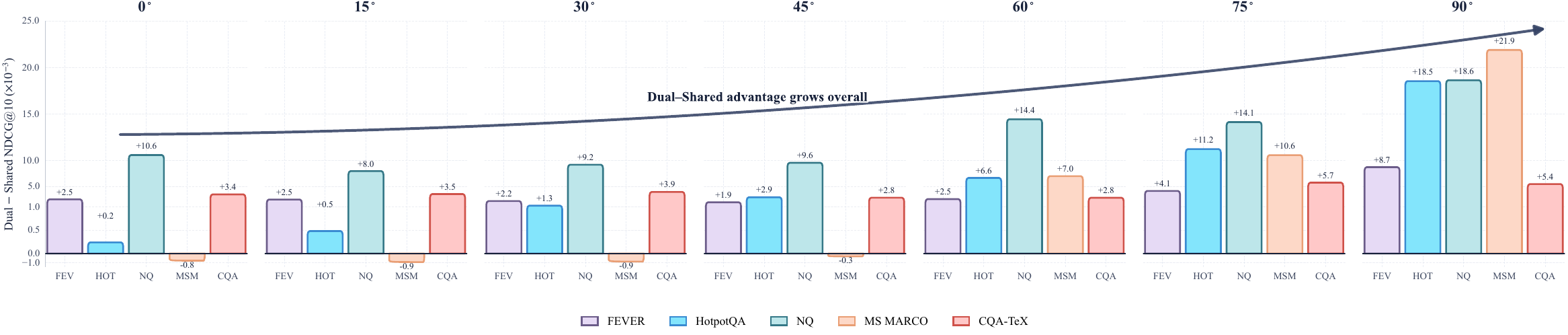}
\caption{Dual-minus-Shared full-corpus NDCG@10 under query rotation.
Each bar averages three paired seeds for BGE-base and GTE-base on each dataset.}
\label{fig:controlled-real}
\end{figure}

Figure~\ref{fig:rank-sweep} shows a clear effect of training size.
Shared wins 13 of the 16 displayed cells at $n=32$, whereas Dual wins
all 32 cells at $n=1024$ and $2048$. Within each grid, rank and rotation
are fixed along the sample-size axis, so the initially Shared-favored
conditions switch to Dual as training data increase. Averaged over the four displayed grids and seven training sizes, rank 32 yields slightly higher absolute NDCG@10 than rank 4 for
both Shared (from 0.7242 to 0.7257) and Dual (from 0.7282 to 0.7289),
without enlarging Dual's relative advantage.

\begin{figure}[htbp]
\centering
\includegraphics[width=\linewidth]{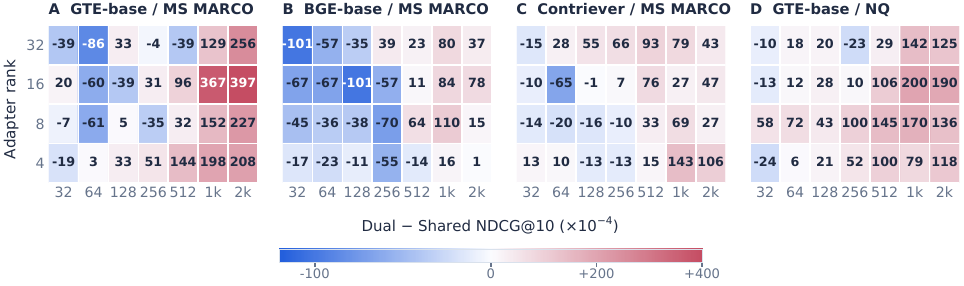}
\caption{Rank--sample-size retrieval gaps. Each cell is mean Dual-minus-Shared test NDCG@10 over three paired seeds, in units of $10^{-4}$; blue favors Shared and red favors Dual.}
\label{fig:rank-sweep}
\end{figure}

\subsection{RQ4: Does operator risk shift with more data, and can CARS select the better geometry?}

The local boundary predicts that, at fixed rank and signal,
the variance penalty for Dual diminishes as the number of training queries
increases.
Across FEVER, NQ, ArguAna, and SciFact, all 168 comparable sample-size slopes
are positive (Appendix Figures~\ref{fig:operator-all-curves}
and~\ref{fig:operator-all-curves-second}). Figure~\ref{fig:gpu-operator}
shows the encoder means for NQ and SciFact. On NQ, three of the four curves
cross from Shared-favored to Dual-favored held-out fit. Together, these
results indicate that more data make directional structure easier to estimate.

\begin{figure}[h]
\centering
\includegraphics[width=\linewidth]{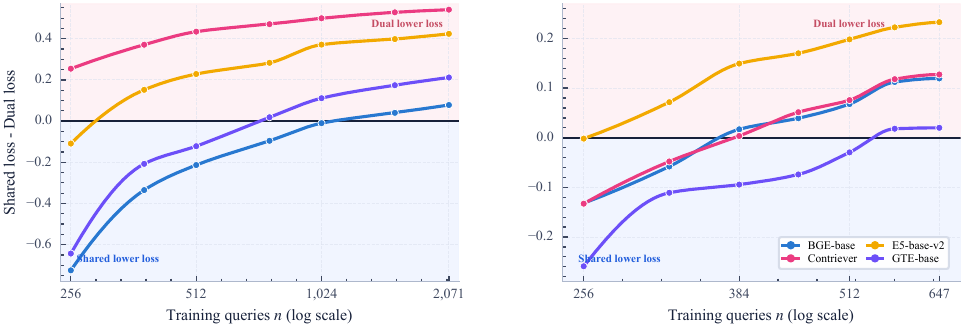}
\caption{Held-out operator fit moves toward Dual with more training queries.
\textbf{A:} NQ; \textbf{B:} SciFact. Each point averages ranks and three seeds
within one encoder; positive values indicate lower held-out loss for Dual.}
\label{fig:gpu-operator}
\end{figure}

\paragraph{Geometry selection.}
We now test the training-only CARS rule in
Equation~\eqref{eq:cars-score} against the fixed Shared and Dual choices.
Table~\ref{tab:ucgs-regret} shows consistent improvements across all four
encoders: CARS wins 85/100 encoder--fold comparisons, reaches 90.1\% mean
selection accuracy, and reduces mean regret by 49--96\% relative to the
better fixed choice. This gain answers the geometry-choice part of RQ4:
when the preferred family varies across training sizes, ranks, and datasets,
cross-fitted signal minus disagreement is more reliable than committing to one geometry throughout. 

\begin{table}[htbp]
\centering
\caption{Geometry selection on five datasets and
four encoders. Lower held-out operator-risk regret and higher decision
accuracy or strict outer-fold wins indicate better performance. Regret and
accuracy are averaged over five held-out query folds.}
\label{tab:ucgs-regret}
\scriptsize
\setlength{\tabcolsep}{3pt}
\renewcommand{\arraystretch}{1.2}
\begin{tabular*}{\linewidth}{@{\extracolsep{\fill}}ll*{4}{rrc}}
\toprule
\multirow[c]{2}{*}{\textbf{Rule}} &
\multirow[c]{2}{*}{\textbf{Dataset}} &
\multicolumn{3}{>{\columncolor{blue!8}}c}{\textbf{BGE-base}} &
\multicolumn{3}{>{\columncolor{pink!12}}c}{\textbf{Contriever}} &
\multicolumn{3}{>{\columncolor{orange!9}}c}{\textbf{E5-base-v2}} &
\multicolumn{3}{>{\columncolor{purple!6}}c}{\textbf{GTE-base}} \\
& &
\cellcolor{blue!8}\textbf{Regret} &
\cellcolor{blue!8}\textbf{Accuracy} &
\cellcolor{blue!8}\textbf{Wins} &
\cellcolor{pink!12}\textbf{Regret} &
\cellcolor{pink!12}\textbf{Accuracy} &
\cellcolor{pink!12}\textbf{Wins} &
\cellcolor{orange!9}\textbf{Regret} &
\cellcolor{orange!9}\textbf{Accuracy} &
\cellcolor{orange!9}\textbf{Wins} &
\cellcolor{purple!6}\textbf{Regret} &
\cellcolor{purple!6}\textbf{Accuracy} &
\cellcolor{purple!6}\textbf{Wins} \\
\midrule
\multirow[c]{5}{*}{\shortstack[l]{\textbf{Always}\\\textbf{Shared}}} & ArguAna & 0.0950 & 45.8\% & 0/5 & 0.1310 & 40.0\% & 0/5 & 0.2543 & 30.0\% & 0/5 & 0.1294 & 43.3\% & 0/5 \\
 & MS MARCO & 0.0680 & 75.6\% & 1/5 & 0.3254 & 46.2\% & 0/5 & 0.0839 & 75.0\% & 1/5 & 0.0442 & 79.4\% & 1/5 \\
 & FEVER & 0.2490 & 33.1\% & 1/5 & 0.4959 & 21.9\% & 0/5 & 0.3968 & 25.0\% & 0/5 & 0.2999 & 37.5\% & 1/5 \\
 & NQ & 0.0050 & \cellcolor{green!13}85.0\% & 1/5 & 0.1915 & 40.0\% & 0/5 & 0.1020 & 54.3\% & 1/5 & 0.0232 & 74.3\% & 2/5 \\
 & SciFact & 0.0226 & 63.3\% & 1/5 & 0.0288 & 58.3\% & 0/5 & 0.0202 & 65.0\% & 2/5 & 0.0374 & 65.0\% & 1/5 \\
\midrule[0.4pt]
\multirow[c]{5}{*}{\shortstack[l]{\textbf{Always}\\\textbf{Dual}}} & ArguAna & 0.3342 & 54.2\% & 0/5 & 0.3395 & 60.0\% & 0/5 & 0.1370 & 70.0\% & 0/5 & 0.3310 & 56.7\% & 0/5 \\
 & MS MARCO & 9.5146 & 24.4\% & 0/5 & 2.5724 & 53.8\% & 0/5 & 7.9767 & 25.0\% & 0/5 & 9.6687 & 20.6\% & 0/5 \\
 & FEVER & 0.4271 & 66.9\% & 0/5 & 0.1882 & 78.1\% & 0/5 & 0.2843 & 75.0\% & 0/5 & 0.6001 & 62.5\% & 0/5 \\
 & NQ & 1.5685 & 15.0\% & 0/5 & 0.6073 & 60.0\% & 0/5 & 0.9152 & 45.7\% & 0/5 & 1.5317 & 25.7\% & 0/5 \\
 & SciFact & 0.4445 & 36.7\% & 0/5 & 0.4369 & 41.7\% & 0/5 & 0.5591 & 35.0\% & 0/5 & 0.4899 & 35.0\% & 0/5 \\
\midrule[0.4pt]
\multirow[c]{5}{*}{\textbf{CARS}} & ArguAna & \cellcolor{green!30}\textbf{0.0018} & \cellcolor{green!30}\textbf{94.2\%} & \cellcolor{green!30}\textbf{5/5} & \cellcolor{green!30}\textbf{0.0059} & \cellcolor{green!30}\textbf{94.2\%} & \cellcolor{green!30}\textbf{5/5} & \cellcolor{green!30}\textbf{0.0131} & \cellcolor{green!13}\textbf{88.3\%} & \cellcolor{green!30}\textbf{5/5} & \cellcolor{green!30}\textbf{0.0015} & \cellcolor{green!30}\textbf{95.0\%} & \cellcolor{green!30}\textbf{5/5} \\
 & MS MARCO & \cellcolor{green!13}\textbf{0.0207} & \cellcolor{green!30}\textbf{93.8\%} & \cellcolor{green!13}\textbf{4/5} & \cellcolor{green!30}\textbf{0.0036} & \cellcolor{green!30}\textbf{96.9\%} & \cellcolor{green!30}\textbf{5/5} & \cellcolor{green!13}\textbf{0.0251} & \cellcolor{green!30}\textbf{93.1\%} & \cellcolor{green!13}\textbf{4/5} & \cellcolor{green!13}\textbf{0.0111} & \cellcolor{green!30}\textbf{91.9\%} & \cellcolor{green!13}\textbf{4/5} \\
 & FEVER & \cellcolor{green!13}\textbf{0.1475} & \cellcolor{green!13}\textbf{87.5\%} & \cellcolor{green!13}\textbf{4/5} & \cellcolor{green!13}\textbf{0.0749} & \cellcolor{green!13}\textbf{88.1\%} & \cellcolor{green!13}\textbf{4/5} & \cellcolor{green!13}\textbf{0.1043} & \cellcolor{green!13}\textbf{86.9\%} & \cellcolor{green!13}\textbf{4/5} & \textbf{0.2671} & \textbf{77.5\%} & \cellcolor{green!13}\textbf{4/5} \\
 & NQ & \textbf{0.0040} & \cellcolor{green!30}\textbf{90.7\%} & \cellcolor{green!13}\textbf{4/5} & \cellcolor{green!30}\textbf{0.0013} & \cellcolor{green!30}\textbf{97.9\%} & \cellcolor{green!30}\textbf{5/5} & \cellcolor{green!13}\textbf{0.0568} & \cellcolor{green!13}\textbf{81.4\%} & \cellcolor{green!13}\textbf{4/5} & \cellcolor{green!13}\textbf{0.0090} & \cellcolor{green!30}\textbf{92.9\%} & \textbf{3/5} \\
 & SciFact & \cellcolor{green!30}\textbf{0.0032} & \cellcolor{green!30}\textbf{92.5\%} & \cellcolor{green!13}\textbf{4/5} & \cellcolor{green!30}\textbf{0.0042} & \cellcolor{green!30}\textbf{93.3\%} & \cellcolor{green!30}\textbf{5/5} & \textbf{0.0193} & \cellcolor{green!13}\textbf{86.7\%} & \textbf{3/5} & \textbf{0.0287} & \cellcolor{green!13}\textbf{80.0\%} & \cellcolor{green!13}\textbf{4/5} \\
\bottomrule
\end{tabular*}
\vspace{2pt}
\end{table}

\section{Discussion and future directions}\label{sec:limitations}

The value of separate projections depends on how query and document
representations align after accounting for their marginal covariances.
The rotation experiment shows a larger Dual advantage as mismatch
increases, while the rank--sample-size grids show that more training
data can reverse a low-sample preference for Shared. Geometry choice
should therefore account for both directional mismatch and available
training data.

Our exact risk boundary assumes local, isotropic Gaussian operator noise. Extending risk estimation to anisotropic noise and directly to ranking metrics would make geometry selection more closely reflect retrieval performance. Another direction is partial sharing, with the number of separately parameterized directions selected alongside rank.

\section{Conclusion}

We characterized the approximation and estimation costs of shared and dual projections for dense retrieval. In the local Gaussian model, dual projections have lower risk when squared directional signal exceeds \(\sigma^2r(2p-r-1)/2\). Retrieval and held-out operator experiments show how mismatch, rank, and sample size affect this tradeoff. CARS uses cross-fitted estimates to select between the two projection families, reducing held-out regret by 49--96\% relative to the better fixed choice across five datasets. Finally, this framework may guide data-efficient retrieval-head design and other two-view representation problems in which the inputs play different roles.

\clearpage
\section*{Reproducibility Statement}
Detailed derivations and proofs are provided in the appendix, together with experimental protocols. Randomized experiments use predefined seeds for data splits, training, and cross-fitting; results are averaged across multiple seeds where applicable.

\section*{AI Use Statement}
Generative AI tools were used for language polishing, checking proofs, brainstorming, and reviewing/debugging portions of the experimental code. All AI-generated or AI-modified content has been reviewed and verified by the authors. The authors take full responsibility for all final proofs, code, analyses, and claims.

\bibliography{references}
\bibliographystyle{iclr2027_conference}

\clearpage
\appendix
\section*{Appendix}
The appendix provides proofs, protocols for the reported experiments, and
the held-out operator-risk trajectories cited in the main text.
\section{Complete Proof of the Approximation Theorem}
\label{app:projection}

We first record the strict expressivity relation.  For every
$P\in\mathbb R^{r\times p}$, $P^\top P$ is symmetric,
$v^\top P^\top Pv=\norm{Pv}_2^2\ge0$, and
$\rank(P^\top P)\le r$.  Conversely, if
$M=U\Sigma V^\top$ has rank $q\le r$, define
\[
A=\begin{bmatrix}\Sigma^{1/2}U^\top\\0_{(r-q)\times p}\end{bmatrix},\qquad
B=\begin{bmatrix}\Sigma^{1/2}V^\top\\0_{(r-q)\times p}\end{bmatrix}.
\]
Then $A^\top B=M$. For $1\le r\le p$, the rank-one matrix
$-e_1e_1^\top$ belongs to $\Md(r)$ but not $\Ms(r)$, so
$\Ms(r)\subsetneq\Md(r)$.

\begin{proof}[Proof of Theorem~\ref{thm:projection}]
For any symmetric $S$, the Frobenius inner product between the skew-symmetric
$K$ and symmetric $H-S$ vanishes.  Hence
\[
\norm{M_\star-S}_F^2=\norm K_F^2+\norm{H-S}_F^2.
\]
The first term is unavoidable.  Diagonalize
$H=U\operatorname{diag}(\eta)U^\top$.  By orthogonal invariance and the
Hoffman--Wielandt inequality, an optimal symmetric $S$ can be chosen with the
same eigenvectors.  PSD constrains its eigenvalues $s_j$ to be nonnegative, and
the rank constraint permits at most $r$ nonzero values.  A negative $\eta_j$ is
therefore optimally paired with zero.  For a positive $\eta_j$, retaining it at
$s_j=\eta_j$ costs zero and dropping it costs $\eta_j^2$.  The best support
keeps the largest at most $r$ positive eigenvalues.  Adding the orthogonal skew
residual yields Equation~\eqref{eq:projection}.

For comparison with the Dual class, let $M_{\star,r}$ be the rank-$r$ SVD
truncation of $M_\star$ and set
$a_r^2=\sum_{j\le r}\sigma_j(M_\star)^2$. For any $D$ of rank at most $r$,
von Neumann's inequality followed by Cauchy--Schwarz gives
$\langle M_\star,D\rangle_F\le a_r\|D\|_F$. Consequently,
\[
\|M_\star-D\|_F^2
\ge\|M_\star\|_F^2+\|D\|_F^2-2a_r\|D\|_F
\ge\|M_\star\|_F^2-a_r^2
=\sum_{j>r}\sigma_j(M_\star)^2.
\]
Both inequalities are equalities at $D=M_{\star,r}$. Since
$\Ms(r)\subset\Md(r)$, subtracting this Dual minimum from
Equation~\eqref{eq:projection} gives a nonnegative and exact approximation
gap between the two families.
\end{proof}

\section{Geometry and Proof of the Local Risk Theorem}
\label{app:local}

\subsection{Tangent spaces and dimensions}
For $j\in\{s,d\}$, write $d_j=\dim T_j$.
Let $S=U\Lambda U^\top$ with $U\in\mathbb R^{p\times r}$ orthonormal and
$\Lambda\succ0$.  Extend $U$ to an orthogonal matrix $[U,U_\perp]$.

For the general fixed-rank manifold, first-order perturbation of an SVD gives
\begin{equation}
T_d=\left\{
UAU^\top+U_\perp BU^\top+UCU_\perp^\top:
\begin{array}{l}
A\in\mathbb R^{r\times r},\\
B\in\mathbb R^{(p-r)\times r},\\
C\in\mathbb R^{r\times(p-r)}
\end{array}\right\}.
\label{eq:tdetail}
\end{equation}
The three blocks are orthogonal and free, so
$\dim T_d=r^2+2r(p-r)=2pr-r^2$.

For the PSD manifold, use the local representation
$S(t)=U(t)\Lambda(t)U(t)^\top$.  Differentiating at zero gives
\begin{equation}
T_s=\left\{
UAU^\top+U_\perp BU^\top+UB^\top U_\perp^\top:
A=A^\top,\ B\in\mathbb R^{(p-r)\times r}
\right\}.
\label{eq:tsdetail}
\end{equation}
Therefore
\[
\dim T_s=\frac{r(r+1)}2+r(p-r)
=pr-\frac{r(r-1)}2.
\]
Every matrix in Equation~\eqref{eq:tsdetail} is a special case of
Equation~\eqref{eq:tdetail}, proving $T_s\subset T_d$.  Subtraction gives
$k=r(2p-r-1)/2$.

The difference space also has an explicit orthogonal block description. In
the basis $[U,U_\perp]$, write $X\in T_d$ as
$\left[\begin{smallmatrix}A&C\\B&0\end{smallmatrix}\right]$.
Its $T_s$ component has upper-left block $\sym(A)$ and paired cross-blocks
$B_s=(B+C^\top)/2$ and $B_s^\top$. The remaining component in
$\mathcal A=T_d\cap T_s^\perp$ is
\[
\begin{bmatrix}
\operatorname{skew}(A)&-D^\top\\D&0
\end{bmatrix},\qquad D=\frac{B-C^\top}{2}.
\]
The skew block contributes $r(r-1)/2$ coordinates and the independent
cross-block $D$ contributes $r(p-r)$, proving the stated decomposition of
$k$. Orthogonality also gives
$\|\Pi_{\mathcal A}X\|_F^2=\|\operatorname{skew}(A)\|_F^2+
\tfrac12\|B-C^\top\|_F^2$.

The corresponding orthogonal projectors can be written explicitly.  With
$P=UU^\top$ and $Z_{\rm sym}=(Z+Z^\top)/2$,
\begin{align*}
\Pi_{T_d}(Z)&=PZ+ZP-PZP,\\
\Pi_{T_s}(Z)&=PZ_{\rm sym}+Z_{\rm sym}P-PZ_{\rm sym}P.
\end{align*}

\subsection{Derivative of nearest-point projection}

We give a self-contained first-order argument.  Let $\mathcal N$ be either
regular manifold and choose a smooth local chart
$\phi(u)=S+Lu+Q(u)$, where $L$ is injective,
$\operatorname{range}(L)=T_S\mathcal N$, and $\norm{Q(u)}=O(\norm u^2)$.
For small ambient perturbation $e$, a locally nearest point
$\phi(\widehat u)$ minimizes $\norm{S+e-\phi(u)}^2$.  Its normal equation is
\[
L^\top(e-L\widehat u)=O(\norm e^2),
\]
because $\widehat u=O(\norm e)$ and both the chart remainder and derivative
remainder are second order.  Solving gives
\[
L\widehat u=L(L^\top L)^{-1}L^\top e+O(\norm e^2)
=\Pi_{T_S\mathcal N}e+O(\norm e^2).
\]
Consequently the nearest-point map satisfies
\begin{equation}
\pi_{\mathcal N}(S+e)=S+\Pi_{T_S\mathcal N}e+O(\norm e^2).
\label{eq:nearestderivative}
\end{equation}
Regularity and $\Lambda\succ0$ ensure a sufficiently small neighborhood with a
locally unique projection.
The global nearest points used in the theorem exist because $\Md(r)$ and
$\Ms(r)$ are closed. Since $S$ belongs to both, any nearest point
$\widehat M_j$ satisfies $\|\widehat M_j-S\|_F\le2\|e\|_F$. For sufficiently
small $e$, the smallest positive singular value of $S$ therefore keeps a
nearest point at rank $r$. It then agrees with the smooth local projection
in Equation~\eqref{eq:nearestderivative}.

\begin{proof}[Proof of Theorem~\ref{thm:localrisk}]
Set $e_n=n^{-1/2}(H+\sigma G)+O(n^{-1})$ in
Equation~\eqref{eq:nearestderivative}.  For $j\in\{s,d\}$,
\[
\widehat M_j
=S+\frac1{\sqrt n}\Pi_{T_j}(H+\sigma G)+O_p(n^{-1}).
\]
Subtracting $M_n=S+n^{-1/2}H+O(n^{-1})$ yields
\[
\sqrt n(\widehat M_j-M_n)
=-\Pi_{T_j^\perp}H+\sigma\Pi_{T_j}G+o_p(1).
\]
The deterministic bias is normal to the projected Gaussian noise.  In any
Frobenius-orthonormal basis of $T_j$, the latter has $d_j$ independent standard
normal coordinates, hence expected squared norm $d_j$. To pass from
convergence in probability to risk convergence, use the global nearest-point
bound $\|\widehat M_j-S\|_F\le2\|Y_n-S\|_F$. The deterministic remainder in
$M_n=S+n^{-1/2}H+O(n^{-1})$ is uniformly bounded after multiplication by
$\sqrt n$, so, for all sufficiently large $n$,
\[
n\|\widehat M_j-M_n\|_F^2
\le C\bigl(1+\|H\|_F^2+\sigma^2\|G\|_F^2\bigr)
\]
for a constant $C$ independent of $n$. The Gaussian right-hand side is
integrable; dominated convergence (after the local expansion, which holds
almost surely for each fixed $G$) therefore gives
\[
n\E\norm{\widehat M_j-M_n}_F^2
\longrightarrow
\norm{\Pi_{T_j^\perp}H}_F^2+\sigma^2d_j.
\]
Since $H\in T_d$, the dual bias vanishes.  Since $T_s\subset T_d$, the shared
bias is its component in $T_d\cap T_s^\perp$.
\end{proof}

\begin{proof}[Proof of Corollary~\ref{cor:phase}]
By the orthogonal decomposition $T_d=T_s\oplus\mathcal A$, the shared
asymptotic risk exceeds the dual asymptotic risk by
\[
\|\Pi_{T_s^\perp}H\|_F^2+\sigma^2d_s-\sigma^2d_d
=\|\Pi_{\mathcal A}H\|_F^2-\sigma^2k
=\delta^2-\sigma^2k.
\]
The difference is positive exactly under Equation~\eqref{eq:oracleboundary};
equality gives a first-order tie. For
$\Delta_n=n^{-1/2}H+O(n^{-1})$, linearity and boundedness of the orthogonal
projector give
$n\|\Pi_{\mathcal A}\Delta_n\|_F^2=\delta^2+O(n^{-1/2})$.
Hence the unscaled expression is the same first-order boundary along these
local alternatives, away from equality.
\end{proof}

\paragraph{Why the theorem is local.}
At rank-changing points the two sets are stratified rather than smooth, so a
tangent cone replaces the tangent space.  For a fixed target far outside the
PSD manifold, the derivative of its nearest-point projection also contains a
curvature (shape-operator) term.  The local alternative avoids both issues and
is the regime in which effective dimension has an exact first-order meaning.

\section{Proof of the SURE Selector}
\label{app:sure}

\begin{proof}[Proof of Theorem~\ref{thm:sure}]
Work in a Frobenius-orthonormal basis of $T_d$, so the observation is an
ordinary $d_d$-dimensional Gaussian sequence.  For the linear estimator
$\widehat H_j=\Pi_{T_j}Z$, Stein's formula \citep{stein1981sure} is
\[
\widehat R_j=\norm{(I-\Pi_{T_j})Z}^2+2\sigma^2d_j-\sigma^2d_d.
\]
Its unbiasedness is also immediate here without a general Stein identity:
$\E\|(I-\Pi_{T_j})Z\|_F^2=
\|(I-\Pi_{T_j})H\|_F^2+\sigma^2(d_d-d_j)$, so
$\E\widehat R_j=\|(I-\Pi_{T_j})H\|_F^2+\sigma^2d_j
=\E\|\widehat H_j-H\|_F^2$.
For $j=d$, this reduces to $\sigma^2d_d$.  Because
$T_d=T_s\oplus\mathcal A$ orthogonally, for $j=s$ it is
\[
\widehat R_s=\norm{\Pi_{\mathcal A}Z}^2+2\sigma^2d_s-\sigma^2d_d.
\]
Thus $\widehat R_d<\widehat R_s$ exactly when
$\norm{\Pi_{\mathcal A}Z}^2>2\sigma^2(d_d-d_s)$, proving
Equation~\eqref{eq:surethreshold}.

In an orthonormal basis of $\mathcal A$,
$\Pi_{\mathcal A}Z/\sigma$ is a $k$-variate identity-covariance normal with
mean squared norm $\lambda=\delta^2/\sigma^2$.  Its squared norm therefore has
the noncentral $\chi_k^2(\lambda)$ distribution, proving
Equation~\eqref{eq:power}.

The two unconditional fixed-model risks in the scaled experiment are
\[
R_s(H)=\delta^2+\sigma^2d_s,
\qquad R_d(H)=\sigma^2d_d.
\]
Their difference is $\delta^2-\sigma^2k$.  Under the model-choice loss defined
in Theorem~\ref{thm:sure}, if this difference is negative, regret occurs only
when SURE selects dual; if positive, regret occurs only when SURE selects
shared.  Multiplication by the corresponding exact tail probability proves
Equation~\eqref{eq:regret}. Replacing $\sigma^2$ by $\sigma^2/n$ gives the
unscaled threshold.
\end{proof}

\subsection{Expectation of the cross-fitted CARS score}
\label{app:cars-score}

The real-embedding selector repeatedly splits data into disjoint halves rather than
the isotropic Gaussian observation in Theorem~\ref{thm:sure}. Its score has a
simple expectation under the stated independent-error model. Write
$V_j=\Delta+\varepsilon_j$, where the two errors are independent and
mean-zero with covariance $2\Sigma/n$ in an orthonormal operator-coordinate
basis. Independence and centering give
\[
\E\langle V_1,V_2\rangle_F=\|\Delta\|_F^2,
\qquad
\E\|V_1-V_2\|_F^2
=\E\|\varepsilon_1\|_F^2+\E\|\varepsilon_2\|_F^2
=\frac{4\operatorname{tr}(\Sigma)}{n}.
\]
Thus
\[
\E\left[\langle V_1,V_2\rangle_F
-\tfrac14\|V_1-V_2\|_F^2\right]
=\|\Delta\|_F^2-\frac{\operatorname{tr}(\Sigma)}{n}.
\]
Averaging over repeated half-splits leaves this expectation unchanged;
independence across the repeated splits is unnecessary for this identity.
If the residuals span the $k$ dual-only directions and
$\Sigma=\sigma^2I_k$, then $\operatorname{tr}(\Sigma)=\sigma^2k$, matching
the unscaled first-order boundary in Corollary~\ref{cor:phase}. The identity
is an expectation calculation; unlike the Gaussian SURE theorem, it does
not assert an exact finite-sample selection probability.

\section{Optimal Rank-$r$ Separation in Two-View Retrieval}
\label{app:deflection}

Let $(X,Y^+)$ be a centered query--positive-document pair with
positive-definite covariances $\Sigma_x,\Sigma_y$. Let $Y^-$ be a centered
document with covariance $\Sigma_y$, independent of $X$. Define the
population relevance moment and its whitened form by
\[
C=\E[X(Y^+-Y^-)^\top]=\E[XY^{+\top}],
\qquad T=\Sigma_x^{-1/2}C\Sigma_y^{-1/2}.
\]
For $j\in\{s,d\}$, the optimal positive--negative separation under a
unit negative-score second moment is
\[
\mathfrak D_j^\star(r)=
\sup_{\substack{M\in\mathcal M_j(r)\\
\E[(X^\top MY^-)^2]\le1}}
\E[X^\top M(Y^+-Y^-)].
\]
Independence gives
$\E[(X^\top MY^-)^2]=\|\Sigma_x^{1/2}M\Sigma_y^{1/2}\|_F^2$.

\begin{theorem}[Optimal rank-$r$ separation]\label{thm:deflection}
Let $\sigma_j(T)$ be the $j$th largest singular value of $T$. Then
\begin{equation}
[\mathfrak D_d^\star(r)]^2=\sum_{j=1}^r\sigma_j(T)^2.
\label{eq:defldual}
\end{equation}
An optimum is proportional to $\Sigma_x^{-1/2}T_r\Sigma_y^{-1/2}$,
where $T_r$ is the rank-$r$ SVD truncation of $T$.
If both views are whitened, so $\Sigma_x=\Sigma_y=I$ and $T=C$, let
$\eta_1\ge\cdots\ge\eta_p$ be the eigenvalues of
$\sym(T)=(T+T^\top)/2$ and $\eta_j^+=\max(\eta_j,0)$. Then
\begin{equation}
[\mathfrak D_s^\star(r)]^2=\sum_{j=1}^r(\eta_j^+)^2.
\label{eq:deflshared}
\end{equation}
\end{theorem}

For whitened views, the squared separation gap
$[\mathfrak D_d^\star(r)]^2-[\mathfrak D_s^\star(r)]^2$ equals the exact
Shared-minus-Dual approximation gap in Theorem~\ref{thm:projection} with
$M_\star=T$.

\begin{corollary}[When shared geometry is sufficient]
\label{cor:matched-views}
If the views are whitened and $T$ is symmetric PSD, then
$\mathfrak D_s^\star(r)=\mathfrak D_d^\star(r)$ for every $r$.
The same equality holds without preprocessing whitening for matched latent
views $X=C_0Z+\epsilon_q$ and $Y^+=C_0Z+\epsilon_d$, where the components
are centered and mutually independent, $\operatorname{Cov}(Z)=I$, and
$\operatorname{Cov}(\epsilon_q)=\operatorname{Cov}(\epsilon_d)$.
\end{corollary}

\begin{proof}[Proof of Theorem~\ref{thm:deflection}]
Because $Y^-$ is centered and independent of $X$,
\[
\E[X^\top M(Y^+-Y^-)]=\tr(M^\top C).
\]
Independence also gives
\begin{align*}
\E[(X^\top MY^-)^2]
&=\E_X\left[X^\top M\Sigma_yM^\top X\right]\\
&=\tr(M^\top\Sigma_xM\Sigma_y).
\end{align*}
Let $N=\Sigma_x^{1/2}M\Sigma_y^{1/2}$.  Then rank is preserved,
the noise constraint is $\norm N_F\le1$, and
\[
\tr(M^\top C)=\ip{N}{T}_F,
\qquad T=\Sigma_x^{-1/2}C\Sigma_y^{-1/2}.
\]
By von Neumann's trace inequality and Cauchy--Schwarz, for rank at most $r$,
\[
\ip{N}{T}_F
\le\left(\sum_{j\le r}\sigma_j(T)^2\right)^{1/2}\norm N_F.
\]
If $T_r=0$, both sides vanish and $N=0$ is optimal. Otherwise equality is
attained by $N=T_r/\norm{T_r}_F$, proving Equation~\eqref{eq:defldual}.

For whitened views and a PSD $M$, write $T=H+K$ with $H$ symmetric and $K$
skew-symmetric.  Then $\ip{M}{K}_F=0$.  Let the eigenvalues of $H$ be
$\eta_1\ge\cdots\ge\eta_p$.  Von Neumann's inequality aligns the eigenvectors
of a maximizing PSD $M$ with those of $H$.  The remaining optimization is
\[
\max_{m_j\ge0,\ \norm m_2\le1,\ \norm m_0\le r}
\sum_jm_j\eta_j,
\]
whose value is the Euclidean norm of the largest at most $r$ positive
eigenvalues. This proves Equation~\eqref{eq:deflshared}. Since the PSD
feasible set is contained in the unrestricted rank-$r$ set, the difference
of their squared optima is nonnegative.
\end{proof}

\begin{proof}[Proof of Corollary~\ref{cor:matched-views}]
When $T=T^\top\succeq0$, its singular values are exactly its nonnegative
eigenvalues, and $\sym(T)=T$. Consequently the spectral sums in
Equations~\eqref{eq:defldual} and \eqref{eq:deflshared} coincide for
every $r$. In the matched latent model, independent view noises imply
$C=\E[XY^{+\top}]=C_0C_0^\top$ and
$\Sigma_x=\Sigma_y=\Sigma=C_0C_0^\top+\Psi\succ0$.
Thus $T=\Sigma^{-1/2}C_0C_0^\top\Sigma^{-1/2}$ is symmetric, and for every
$v$, $v^\top Tv=\|C_0^\top\Sigma^{-1/2}v\|_2^2\ge0$. It is therefore PSD.
The common congruence $M\mapsto\Sigma^{1/2}M\Sigma^{1/2}$ preserves rank
and is a bijection on the PSD cone, so the whitened equality also holds
for the original shared operator class. Under unequal view covariances,
this PSD-preserving congruence no longer applies.
\end{proof}

\paragraph{Planar-rotation example.}
For $T=\rho R_\theta$ with $0<\rho\le1$ in two whitened dimensions,
$R_\theta^\top R_\theta=I_2$, so both singular values of $T$ equal $\rho$.
At rank one, Equation~\eqref{eq:defldual} gives
$\mathfrak D_d^\star(1)=\rho$. Moreover,
$\sym(R_\theta)=(R_\theta+R_\theta^\top)/2=\cos\theta\,I_2$.
For $0\le\theta\le\pi/2$, the largest positive eigenvalue of
$\sym(T)$ is $\rho\cos\theta$, and Equation~\eqref{eq:deflshared} gives
$\mathfrak D_s^\star(1)=\rho\cos\theta$. This proves the endpoint and
intermediate-angle claims used to motivate the rotation intervention.

For the latent two-view model with
$\operatorname{Cov}(\epsilon_q)=\Psi_q$ and
$\operatorname{Cov}(\epsilon_d)=\Psi_d$,
\[
\Sigma_x=C_qC_q^\top+\Psi_q,\quad
\Sigma_y=C_dC_d^\top+\Psi_d,\quad
C=C_qC_d^\top,
\]
and hence $T=\Sigma_x^{-1/2}C_qC_d^\top\Sigma_y^{-1/2}$. Its left and right
singular vectors are the query and document canonical directions.  A singular
gap controls plug-in subspace stability through Wedin-type perturbation bounds
\citep{cai2018wedin}; covariance conditioning separately affects the whitening
error \citep{gao2019stochasticcca}.

\section{A Distribution-Free Bound and Its Limitation}
\label{app:global}

\begin{proof}[Proof of Proposition~\ref{prop:global}]
Let $W=XD^\top$, assume $\norm W_F\le R$ almost surely, and constrain
$\norm M_F\le B$.  Conditional on $W_1,\ldots,W_n$, write
$Q=n^{-1}\sum_i\epsilon_iW_i$.  Von Neumann's inequality gives
\begin{align*}
\sup_{\rank(M)\le r,\norm M_F\le B}\ip M Q_F
&=B\left(\sum_{j\le r}\sigma_j(Q)^2\right)^{1/2}\\
&\le B\norm Q_F.
\end{align*}
Jensen's inequality and independence of Rademacher signs imply
\[
\E_\epsilon\norm Q_F
\le\left(\frac1{n^2}\sum_i\norm{W_i}_F^2\right)^{1/2}
\le\frac R{\sqrt n}.
\]
The shared class is a subset of the dual class, so its Rademacher supremum
cannot exceed the dual one. Together these are exactly the inequalities in
Proposition~\ref{prop:global}.

For completeness, let $\ell$ be an $L$-Lipschitz bounded margin loss. The
constant $\ell(0)$ can be subtracted without changing the deviation of the
empirical risk. Symmetrization bounds the expected uniform deviation of the
centered loss class by twice its Rademacher complexity; contraction bounds
that loss complexity by a constant multiple of $L$ times the linear-score
complexity. Applying Proposition~\ref{prop:global} gives
$O(LBR/\sqrt n)$ for either family
\citep{bartlett2002rademacher}.

This argument cannot produce the strict penalty in Equation~\eqref{eq:k}: it
throws away the leading-$r$ spectrum of $Q$ by upper-bounding it with the full
Frobenius norm. More fundamentally, the boundedness assumptions allow a
degenerate design $W_i=0$ (and nonzero designs of arbitrarily small scale), so
they cannot imply a uniformly positive complexity gap that depends only on
$p$ and $r$. The local theorem adds the regular design and Gaussian
localization needed to expose effective dimension.
\end{proof}

\section{Experimental Protocols and Supporting Evidence}
\label{app:experiments}

\paragraph{Retrieval metric.}
For an evaluated query $q$, let $z_{qj}\in\{0,1\}$ indicate relevance
at rank $j$ and let $R_q>0$ be its number of relevant documents. We report
the mean over test queries of
\[
\operatorname{NDCG@10}(q)=
\frac{\sum_{j=1}^{10}z_{qj}/\log_2(j+1)}
{\sum_{j=1}^{\min(10,R_q)}1/\log_2(j+1)}.
\]

\paragraph{Synthetic retrieval and local calibration (RQ1).}
The two-view retrieval simulation uses ambient dimension $p=32$, latent rank
eight, query and document noise standard deviation $0.6$, and master seed
20260809. The document loading rotates relative to the query loading through
angles $0,0.2,\ldots,1.2$ radians. Training sizes are 32--2048 pairs, fitted
ranks are $4/8/16$, and five paired seeds control training and test draws.
Each positive document is ranked against 299 independent negatives.
Figure~\ref{fig:phase} displays the rank-eight boundary from this grid.

The local-risk calibration uses $p=12$, $r=3$, $\sigma=1$, and
$n\in\{64,256,1024,4096\}$. At each of four directional-signal levels,
20,000 nonlinear-projection draws estimate the two risks; 200,000 Gaussian
sequence draws estimate SURE selection probability. Table~\ref{tab:local-calibration}
reports the $n=4096$ comparisons used in RQ1.

\begin{table}[H]
\caption{Local risk and SURE calibration. Entries are simulation/theory; risks are scaled by $n$.}
\label{tab:local-calibration}
\centering
\small
\begin{tabular}{rrrr}
\toprule
$\delta^2/(\sigma^2k)$ & Shared risk & Dual risk & $\Pr(\widehat\pi=d)$ \\
\midrule
0.0 & 32.97/33 & 63.03/63 & 0.001/0.001 \\
0.5 & 48.05/48 & 63.05/63 & 0.094/0.093 \\
1.0 & 62.93/63 & 62.97/63 & 0.473/0.474 \\
2.0 & 93.10/93 & 63.15/63 & 0.970/0.969 \\
\bottomrule
\end{tabular}
\end{table}

\paragraph{Training-score diagnostic (RQ2).}
The rank-eight simulation compares raw plug-in PSD distance with the known
population distance and its uncertainty-corrected estimate. On frozen real
embeddings, Figure~\ref{fig:noise-correction}B compares raw and corrected
training scores in 840 operator fits from 30 available task--encoder pairs.
The ten tasks are Climate-FEVER, CQA-English, CQA-TeX, DBPedia Entity, FEVER,
HotpotQA, MS MARCO, NQ, TREC-COVID, and Touché-2020. Fits vary rank over
$4/8/16/32$, use three seeds and available nested training sizes, and share
the four encoders specified in Section~\ref{sec:experimental-setting}.
Disjoint held-out relevance moments define the Shared-minus-Dual operator-loss
advantage against which the training-only scores are compared.

\subsection{Real-embedding retrieval and operator-risk protocols}
\label{app:gpu-experiments}

\paragraph{Full-corpus query rotation (RQ3).}
The reported rotation experiment uses FEVER, HotpotQA, NQ, MS MARCO, and
CQA-TeX with BGE-base and GTE-base. Their train/validation/test query counts
and corpus sizes are, respectively, 109,810/6,666/6,666 and 5,416,568;
85,000/5,447/7,405 and 5,233,329; 2,071/691/690 and 2,681,468;
502,939/6,980/43 and 8,841,823; and 1,744/581/581 and 68,184.
For each dataset--encoder pair, eight relevance-informed planes define query
rotations at $0^\circ,15^\circ,\ldots,90^\circ$. Documents, relevance labels,
and query splits remain fixed across angles.

Shared and Dual rank-16 adapters use the same 1,024 training queries, three
paired seeds, epoch permutations, in-batch negatives, and eight frozen hard
negatives. AdamW runs for 30 epochs with batch size 256, temperature $0.05$,
weight decay $10^{-5}$, and gradient norm 5. Validation NDCG@10 chooses the
learning rate from $\{3\times10^{-4},10^{-3},3\times10^{-3}\}$ and the
checkpoint for each family. Test scores use exact float32 full-corpus
inner-product search. These are the conditions plotted in
Figure~\ref{fig:controlled-real}.

\paragraph{Rank-by-data retrieval (RQ3).}
The four grids in Figure~\ref{fig:rank-sweep} are MS MARCO with GTE-base,
BGE-base, and Contriever-MSMARCO at $90^\circ$, and NQ with GTE-base at
$75^\circ$. Each grid crosses seven training sizes
$\{32,64,128,256,512,1024,2048\}$ with ranks $\{4,8,16,32\}$ and three
paired seeds. Training lasts 30 epochs at learning rate $3\times10^{-4}$;
validation selects the checkpoint and the test ranks the complete corpus.

\paragraph{Held-out operator risk (RQ4).}
The 168 comparable sample-size trajectories in
Figures~\ref{fig:operator-all-curves} and~\ref{fig:operator-all-curves-second}
come from FEVER, NQ, ArguAna, and SciFact with all four base encoders and
three seeds. SciFact (647/162/300 train/validation/test queries, 5,183
documents) and FEVER use their official splits. NQ and ArguAna use
deterministic adaptation splits of 2,071/691/690 over 2,681,468 documents
and 844/281/281 over 8,674 documents, respectively. One ArguAna test qrel
has no released document, leaving 280 evaluable test queries. The available
ranks are $4/8/16/32$; ArguAna and SciFact additionally use
$64/128/256$.

Encoder-prescribed query and passage prefixes are applied, texts are truncated
at 256 tokens, and normalized 768-dimensional embeddings are cached in
float16. For each query, the first labeled positive and highest-ranked
unlabeled raw candidate form a relevance triple. The held-out triples define
a disjoint target moment. Truncated SVD and positive eigentruncation of the
training moment give the Dual and Shared operators; squared Frobenius errors
are divided by the squared norm of the held-out target. Each trajectory fixes
dataset, encoder, rank, and seed and connects only measured nested sample
sizes.

\begin{figure}[p]
\centering
\includegraphics[width=\linewidth]{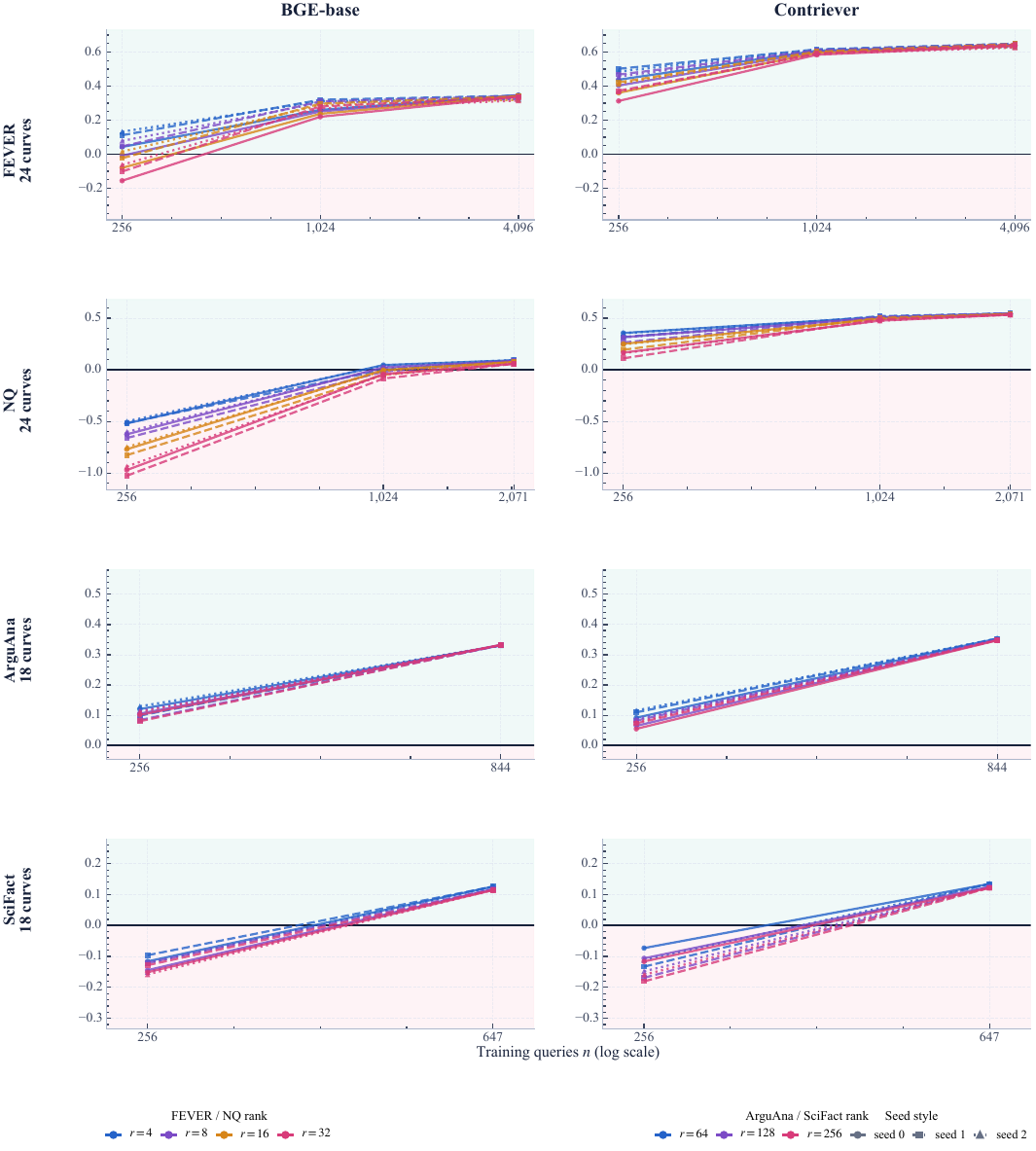}
\caption{Held-out operator-risk trajectories for BGE-base and Contriever
(84 curves). Rows are FEVER and NQ (24 curves each), followed by ArguAna and
SciFact (18 each). Each line fixes a dataset, encoder, rank, and seed and
connects measured nested training sizes. Color denotes rank; line style
denotes seed. Positive normalized Shared-minus-Dual loss favors Dual.}
\label{fig:operator-all-curves}
\end{figure}

\begin{figure}[p]
\centering
\includegraphics[width=\linewidth]{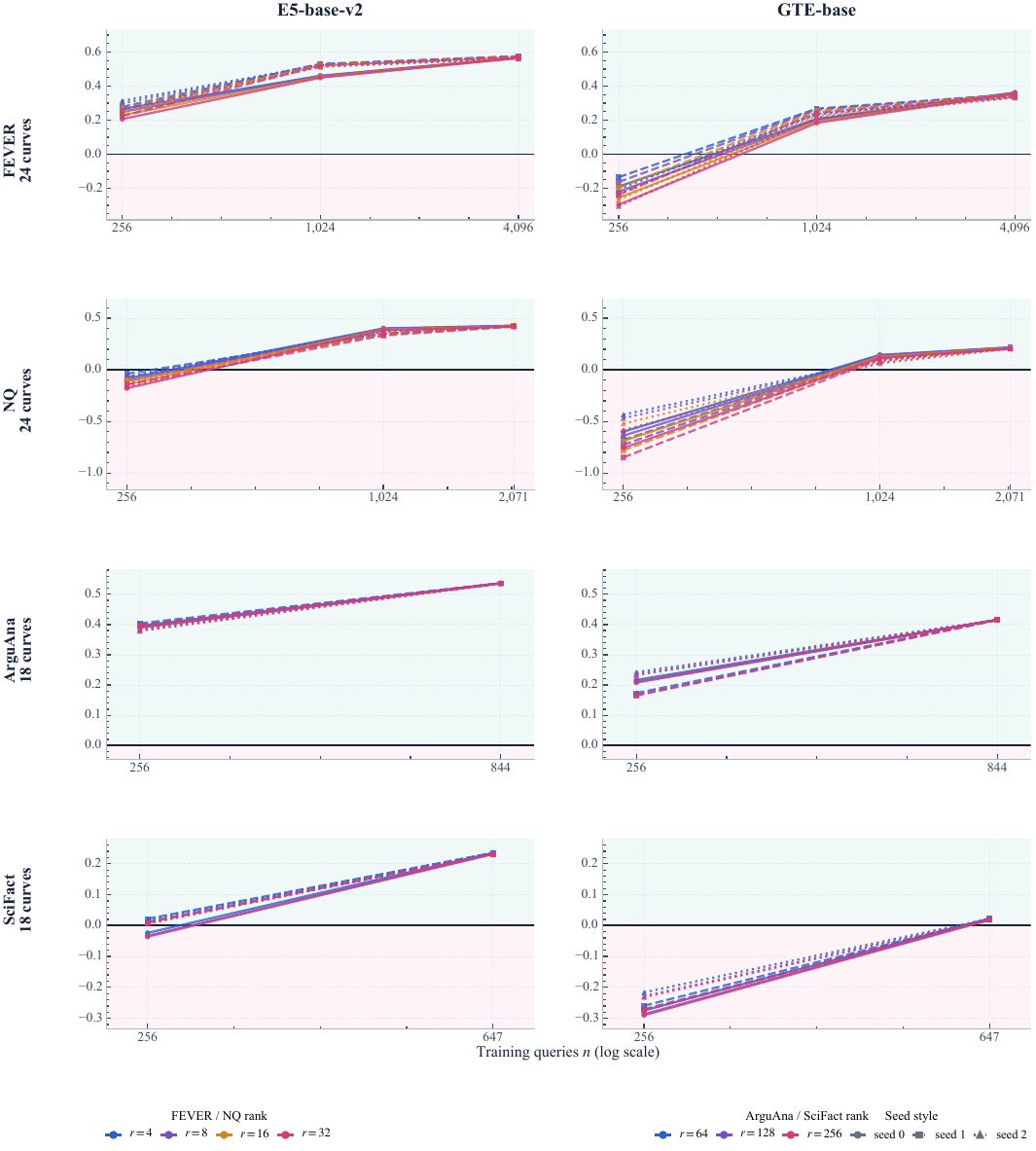}
\caption{The remaining 84 held-out operator-risk trajectories, for E5-base-v2
and GTE-base. Grouping, axes, colors, and line styles match
Figure~\ref{fig:operator-all-curves}.}
\label{fig:operator-all-curves-second}
\end{figure}

\paragraph{Five-fold geometry selection (RQ4).}
The evaluation behind Table~\ref{tab:ucgs-regret} uses ArguAna, MS MARCO,
FEVER, NQ, and SciFact with all four base encoders. For each dataset, labeled
training-query IDs are sorted, permuted with seed 20260923, and partitioned
into five balanced held-out folds shared across encoders. Each fold's target
moment uses only held-out queries; the remaining four folds supply training
queries. A fold-specific permutation with seed $20260924+\text{fold}$
provides nested samples of sizes $32,64,128,256,512,1024,2048,4096$ when
available, plus the full remaining pool when it has at most 4096 queries.
All runs use ranks $4/8/16/32$, a fixed 256-dimensional selector projection,
and 20 internal half-splits. The 2,800 configuration evaluations are nested
within 100 encoder--dataset outer folds; regret is averaged across
sample-size--rank cells, then held-out folds, as specified in
Section~\ref{sec:experimental-setting}.
\section{Additional related work}\label{app:related-work}
Hard-negative mining and data augmentation are central to dense retrievers
\citep{xiong2021ance,qu2021rocketqa}; domain-adaptive and expansion variants
include Contriever, GPL, COCO-DR, and HyDE
\citep{izacard2022contriever,wang2022gpl,yu2022cocodr,gao2023hyde}.
Query-side anisotropy, compression, and frozen-embedding adapters address
different aspects of representation mismatch
\citep{li2021dance,ma2021compression,yoon2024search,yoon2024matryoshka}.
CCA and deep CCA learn two-view directions
\citep{hotelling1936cca,bach2005pcca,andrew2013dcca,dorfer2018cca};
perturbation theory studies their sampling stability
\citep{davis1970rotation,gao2019stochasticcca,cai2018wedin}.
Low-rank and nearest-PSD projections
\citep{eckart1936approximation,higham1988nearest}, fixed-rank manifolds
\citep{vandereycken2013lowrank,vandereycken2013psd}, and risk estimation
\citep{gavish2017optimal,mazumder2020degrees} supply ingredients for our
operator analysis. Pairwise-ranking theory and NDCG address a different
order-based objective \citep{burges2005ranknet,clemencon2009ranking,jarvelin2002ndcg}.

\section{Identifiability and boundary cases}\label{sec:identifiability}
\paragraph{Marginal mismatch is neither necessary nor sufficient.}
If $Y=cX$ with $c>0$ and scoring uses cosine, query and document covariances
differ by $c^2$ while every paired cosine and ranking is unchanged.
Conversely, if $X\sim\mathcal N(0,I)$ and $Y=RX$ for a nonsymmetric
orthogonal rotation $R$, the marginals agree but the cross-covariance is
generally outside the PSD cone. Relevance depends on paired structure.

\paragraph{A bilinear operator does not determine projected dual cosine.}
For invertible $Q$, $(A,B)$ and $(QA,Q^{-\top}B)$ preserve $A^\top B$ but
change projected norms. Take $M=I$, $x=(1,1)$, $d_1=(1,0)$, and
$d_2=(.9,.9)$. With $A=B=I$, projected cosine ranks $d_2$ above $d_1$;
with $A=\operatorname{diag}(10,1)$ and
$B=\operatorname{diag}(.1,1)$, still $A^\top B=I$, but the ranking reverses.
The same invariance changes $\norm{A-B}_F$ arbitrarily. Operator-level
results require bilinear scoring or additional norm control.

\end{document}